\documentclass[11pt]{article}
\usepackage{snapshot}
\usepackage{fullpage}

\usepackage{amsfonts}
\usepackage{amssymb}
\usepackage{amsthm}
\usepackage{amsmath}
\usepackage{mathtools}

\usepackage{ifpdf}
\newcommand{\mydriver}{hypertex}
\ifpdf
 \renewcommand{\mydriver}{pdftex}
\fi
\usepackage[breaklinks,\mydriver]{hyperref}
\newcommand{\R}{\circledR}

\newcommand{\re}{\mathbb{R}}

\newcommand{\vol}{\textrm{vol}}

\newcommand{\Tr}{\textrm{Tr}}

\newcommand{\LL}{{\bf \mathcal{L}}}
\newcommand{\HH}{{\bf \mathcal{H}}}
\newcommand{\supp}{\textrm{supp}}

\newcommand{\p}{\textbf{p}}
\newcommand{\pp}{\tilde{\textbf{p}}}
\newcommand{\q}{\textbf{q}}

\newcommand{\x}{\textbf{x}}
\newcommand{\y}{\textbf{y}}

\newcommand{\vv}{\textbf{v}}
\newcommand{\1}{\textbf{1}}

\newcommand{\reg}{\textrm{reg}}

\newtheorem{lemma}{Lemma}
\newtheorem{theorem}{Theorem}
\newtheorem{corollary}{Corollary}
\newtheorem{claim}[theorem]{Claim}

\theoremstyle{definition}
\newtheorem{definition}{Definition}

 \providecommand{\norm}[1]{\lVert#1\rVert}

\begin{document}
\title{Testing Small Set Expansion in General Graphs\thanks{Both authors are partially supported by the Grand Project ``Network Algorithms and Digital Information'' of the
Institute of software, Chinese Academy of Sciences and by a National Basic Research Program (973) entitled computational Theory on Big Data of Cyberspace, grant No. 2014CB340302. The second author acknowledges the support of ERC grant No. 307696.}}
\author{Angsheng Li\footnote{State Key Laboratory of
Computer Science, Institute of Software, Chinese Academy of Sciences}\\ angsheng@ios.ac.cn
\and Pan Peng\footnote{Department of Computer Science,
Technische Universit{\"a}t Dortmund; State Key Laboratory of
Computer Science, Institute of Software, Chinese Academy of Sciences}\\
pan.peng@tu-dortmund.de}

\date{}
\maketitle

\begin{abstract}
We consider the problem of testing small set expansion for general graphs. A graph $G$ is a $(k,\phi)$-expander if every subset of volume at most $k$ has conductance at least $\phi$. Small set expansion has recently received significant attention due to its close connection to the unique games conjecture, the local graph partitioning algorithms and locally testable codes.

We give testers with two-sided error and one-sided error in the \textit{adjacency list} model that allows degree and neighbor queries to the oracle of the input graph. The testers take as input an $n$-vertex graph $G$, a volume bound $k$, an expansion bound $\phi$ and a distance parameter $\varepsilon>0$. For the two-sided error tester, with probability at least $2/3$, it accepts the graph if it is a $(k,\phi)$-expander and rejects the graph if it is $\varepsilon$-far from any $(k^*,\phi^*)$-expander, where $k^*=\Theta(k\varepsilon)$ and $\phi^*=\Theta(\frac{\phi^4}{\min\{\log(4m/k),\log n\}\cdot(\ln k)})$. The query complexity and running time of the tester are $\widetilde{O}(\sqrt{m}\phi^{-4}\varepsilon^{-2})$, where $m$ is the number of edges of the graph. For the one-sided error tester, it accepts every $(k,\phi)$-expander, and with probability at least $2/3$, rejects every graph that is $\varepsilon$-far from $(k^*,\phi^*)$-expander, where $k^*=O(k^{1-\xi})$ and $\phi^*=O(\xi\phi^2)$ for any $0<\xi<1$. The query complexity and running time of this tester are $\widetilde{O}(\sqrt{\frac{n}{\varepsilon^3}}+\frac{k}{\varepsilon \phi^4})$.

We also give a two-sided error tester in the \textit{rotation map} model that allows \textit{(neighbor, index)} queries and degree queries. This tester has asymptotically almost the same query complexity and running time as the two-sided error tester in the adjacency list model, but has a better performance: it can distinguish any $(k,\phi)$-expander from graphs that are $\varepsilon$-far from $(k^*,\phi^*)$-expanders, where $k^*=\Theta(k\varepsilon)$ and $\phi^*=\Theta(\frac{\phi^2}{\min\{\log(4m/k),\log n\}\cdot(\ln k)})$.

In our analysis, we introduce a new graph product called \textit{non-uniform replacement product} that transforms a general graph into a bounded degree graph, and approximately preserves the expansion profile as well as the corresponding spectral property.
\end{abstract}

\section{Introduction}
Graph property testing is an effective algorithmic paradigm to deal with real-world networks, the scale of which has become so large that it is even impractical to read the whole input. In the setting of testing a graph property $P$, we are given as input a graph $G$ and we want to design an algorithm (called \textit{tester}) to distinguish the case that $G$ has property $P$ from the case that $G$ is ``far from'' the property $P$ with high success probability (say $2/3$). Here, the notion of being ``far from'' is parameterized by a distance parameter $\varepsilon$. In most situations, a graph $G$ is said to be $\varepsilon$-far from property $P$ if one has to modify at least an $\varepsilon$ fraction of the representation (or edges) of $G$ to obtain a graph $G'$ with property $P$. We assume the input graph $G$ can be accessed through an oracle $\mathcal{O}_G$ and the goal is to design property testers that make as few queries as possible to $\mathcal{O}_G$.

Since the seminal work of Goldreich and Ron~\cite{GR02:testing}, many testers have been developed for different graph properties, such as $k$-colorability, bipartiteness, acyclicity, triangle-freeness and many others. Most of these testers apply only to the adjacency matrix model or the adjacency list model, depending on the types of queries the testers are allowed to ask the oracle. The former model is most suitable for \textit{dense} graphs and general characterizations on the testability of a property in this model has been given (e.g., \cite{AFNS09:test}). The latter model is most suitable for \textit{sparse} graphs, and several property testers using the techniques of local search or random walks are known, while it is not well understood what properties are testable in constant time in this model. Even less is known about testers, testability results or even models for general graphs~(see recent surveys~\cite{Ron10:testing,Gol10:test}).

In this paper, we focus on property testers for general graphs. We will consider the \textit{adjacency list} model that allows \textit{degree} queries and \textit{neighbor} queries to the oracle of the graph~\cite{PR02:diameter}. For the degree query, when specified a vertex $v$, the oracle returns the degree of $v$; for the neighbor query, when specified a vertex $v$ and an index $i$; the oracle returns the $i$th neighbor of $v$. The adjacency list model also applies to the bounded degree graphs with an additional restriction that a fixed upper-bound was assumed on the degrees~\cite{GR02:testing}. We will also consider a new model which we call \textit{rotation map} model that allows degree queries and \textit{(neighbor, index)} queries to the oracle~\cite{LPP11:conductance}. For the (neighbor, index) query, when specified a vertex $v$ and an index $i$, the oracle returns a pair $(u,j)$ such that $u$ is the $i$th neighbor of $v$ and $j$ is the index of $u$ as a neighbor of $v$. Note that the rotation map model is at least as strong as the adjacency list model.

We study the problem of testing \textit{small set expansion} for general graphs. Given a graph $G=(V,E)$ with $n$ vertices and $m$ edges, and a set $S\subseteq V$, let the \textit{volume} of $S$ be the sum of degree of vertices in $S$, that is, $\vol(S):=\sum_{v\in S} \deg_G(v),$ where $\deg_G(v)$ denotes the degree of vertex $v$.
Define the \textit{conductance} of $S$ as
$\phi(S):=\frac{e(S,V\backslash S)}{\vol(S)},$ where $e(S,V\backslash S)$ is the number of edges leaving $S$; and define the \textit{$k$-expansion profile} of $G$ as $\phi(k):=\min_{S:\vol(S)\leq k}\phi(S)$.
A graph $G$ is called a $(k,\phi)$-expander if $\phi(k)\geq \phi$, that is, all the subsets in $G$ with volume at most $k$ have conductance at least $\phi$. We will refer to small set expander as $(k,\phi)$-expander and refer to small set expansion as $\phi(k)$.

Besides of the relation to the mixing time of random walks~\cite{LK99:mixing}, small set expansion has been of much interest recently for its close connection to the unique games conjecture~\cite{RS10:expansion,ABS10:subexp}, the design of local graph partitioning algorithms in massive graphs~\cite{ST08:local,ACL06:localpr,AP09:evolvingset,OT12:clustering,KL12:sparse}, and locally testable codes that are testable with linear number of queries~\cite{BGHMRS12:longcode}. Approximation algorithms and spectral characterizations for the small set expansion problem have been studied~\cite{ABS10:subexp,LRTV12:sparse,LOT12:spectral,KL12:sparse,OT12:clustering,OW12:sse}. It is natural to ask if one can efficiently (in sublinear time) test if a graph is a small set expander.


\subsection{Our results}
We give testers for small set expansion in the adjacency list model as well as the rotation map model for general graphs. We use the common definition of distance between graphs. More precisely, a graph $G$ with $m$ edges is said to be $\varepsilon$-far from a $(k,\phi)$-expander if one has to modify at least $\varepsilon m$ edges of $G$ so that it becomes a $(k,\phi)$-expander. We will assume throughout the paper that $m=\Omega(n)$ (and a brief discussion is given in Section~\ref{sec:prelim}), while the algorithm is not given as input the number of edges $m$.

\subsubsection{Testers in adjacency list model}
Our first result is a property tester for small set expansion with two-sided error in the adjacency list model.

\begin{theorem}~\label{thm:twoside-list}
Given degree and neighbor query access to an $n$-vertex graph, a volume bound $k$, a distance parameter $\varepsilon$ and a conductance bound $\phi$, there exists an algorithm that with probability at least $2/3$, accepts any graph that is a $(k,\phi)$-expander, and rejects any graph that is $\varepsilon$-far from any $(k^*,\phi^*)$-expander, where $k^*=\Theta(k\varepsilon)$ and $\phi^*=\Theta(\frac{\phi^4}{\min\{\log(4m/k),\log n\}\cdot(\ln k)})$, where $m$ is the number of edges of $G$. The query complexity and running time of the algorithm are $\widetilde{O}(\sqrt{m}\phi^{-4}\varepsilon^{-2})$.
\end{theorem}

Note that the running time of the tester matches the best known algorithms for testing the conductance of $G$ which corresponds to the case $k=m$ (see further discussions below).

As a byproduct of our analysis for the above two-sided error tester, we obtain a one-sided error tester (that accepts every $(k,\phi)$-expander) by invoking a local algorithm for finding small sparse cuts. We show the following result.

\begin{theorem}~\label{thm:oneside}
Given degree and neighbor query access to an $n$-vertex graph, a volume bound $k$, a conductance bound $\phi$, and a distance parameter $\varepsilon$, there exists an  algorithm that always accepts any graph that is a $(k,\phi)$-expander, and with probability at least $2/3$ rejects any graph that is $\varepsilon$-far from any $(k^*,\phi^*)$-expander, where $k^*=O(k^{1-\xi})$ and $\phi^*=O(\xi\phi^2)$ for any $0<\xi<1$. Furthermore, whenever it rejects a graph, it provides a \textit{certificate} that the graph is not a $(k,\phi)$-expander in the form of a set of volume at most $k$ and expansion at most $\phi$. The query complexity and running time of the algorithm are $\widetilde{O}(\sqrt{\frac{n}{\varepsilon^3}}+\frac{k}{\varepsilon \phi^4})$.
\end{theorem}

Note that $\xi$ is not necessarily a constant, and the running time of the above algorithm is sublinear in $m$ for $k=O(\frac{m}{\log^{\Omega(1)}n})$ and constant $\phi$.

\subsubsection{Tester in rotation map model}
We also give a two-sided error tester in the rotation map model. Note that the gap of the conductance value in completeness and soundness here is smaller than the corresponding gap in the tester in adjacency list model.
\begin{theorem}~\label{thm:twoside-map}
Given degree and (neighbor, index) query access to an $n$-vertex graph, a volume bound $k$, a distance parameter $\varepsilon$ and a conductance bound $\phi$, there exists an algorithm that with probability at least $2/3$, accepts any graph that is a $(k,\phi)$-expander, and rejects any graph that is $\varepsilon$-far from any $(k^*,\phi^*)$-expander, where $k^*=\Theta(k\varepsilon)$ and $\phi^*=\Theta(\frac{\phi^2}{\min\{\log(4m/k),\log n\}\cdot(\ln k)})$, where $m$ is the number of edges of $G$. The query complexity and running time of the algorithm are $\widetilde{O}(\sqrt{m}\phi^{-2}\varepsilon^{-2})$.
\end{theorem}

\subsubsection{Graph transformation}
The analysis of the above two-sided error tester involves analyzing random walks on a bounded degree graph by the spectral property of small set expander and a new graph product which we call \textit{non-uniform replacement product} that transforms every graph (with possible multiple edges and self-loops) into a bounded degree graph, and in the process, the expansion profile of the resulting graph does not differ by much from that of the original graph. This transformation may be of independent interest, and we present the formal result below. Let $\LL_G$ be the normalized Laplacian matrix of a graph $G$ and let $\lambda_i(G)$ denote the $i$th smallest eigenvalues of $\LL_G$.

\begin{theorem}~\label{thm:spectra}
Let $\phi<1$ and $k\leq m$. For any graph $G=(V,E)$ with $n$ vertices and $m$ edges, there exists a $16$-regular graph $G'$ with $\Theta(m)$ vertices such that
\begin{enumerate}
\item If $S\subseteq V(G)$ is a subset in $G$ with $\phi_G(S)\leq \phi$, then there exists a set $S'\subseteq V(G')$, such that $|S'|=\Theta(\vol_G(S))$ and $\phi_{G'}(S')\leq \phi/16$;
\item If for any set $S\subset V(G)$ with $\vol_G(S)\leq k$, $\phi_G(S)\geq \phi$, then
\begin{enumerate}
\item for any $S'\subseteq V(G')$ with $|S'|\leq \Theta(k)$, $\phi_{G'}(S')= \Omega(\phi^2)$.
\item for any $\alpha>0$, it holds that $\lambda_{\frac{(1+\alpha)2m}{k}}(G') = \Omega(\alpha^6 \phi^2 (\log\frac{2m}{k})^{-1})$, and $\lambda_{(\frac{2m}{k})^{1+\alpha}}(G') = \Omega(\alpha\phi^2\log_n\frac{2m}{k})$. Furthermore, if $k=m$, then $\lambda_2(G')=\Omega(\phi^2)$.
\end{enumerate}
\end{enumerate}
\end{theorem}

Note that by recent spectral characterization of small set expansion of $G$ and the preconditions of the Item 2 of Theorem~\ref{thm:spectra}, we have $\lambda_{\frac{(1+\alpha)2m}{k}}(G) = \Omega(\alpha^6\phi^2(\log\frac{2m}{k})^{-1})$, $\lambda_{(\frac{2m}{k})^{1+\alpha}}(G) = \Omega(\alpha\phi^2\log_n(2m/k))$, and if $k=m$, $\lambda_2(G)=\Omega(\phi^2)$ (see Section~\ref{subsec:spectra}). Also we stress that Item 2b above is not a direct consequence of Item 2a and inequalities in Section~\ref{subsec:spectra}, and its proof involves a more refined spectral analysis. The main point from $G$ to $G'$ is that the property of small set expansion is well preserved and the maximum degree is also greatly reduced, which is comparable to work on constructions from high degree expanders to constant degree expanders~(see eg.,\cite{Rei08:connectivity,ASS08:expander}). 

\subsection{Other related work}
There is an interesting line of research on testing the special case of the $(k,\phi)$-expander for $k=m$, which is often abbreviated as \textit{$\phi$-expander}. The corresponding quantity $\phi(m)$ is often called the \textit{expansion (or conductance)} of $G$~\cite{HLW06:expander}. Goldreich and Ron~\cite{GR00:expansion} have proposed an expansion tester for bounded degree graphs in the adjacency list model. The tester (with different setting parameters) has later been analyzed by Czumaj and Sohler~\cite{CS10:expansion}, Nachmias and Shapira~\cite{NS10:expansion}, and Kale and Seshadhri~\cite{KS11:tester}, and it is proven that the tester can distinguish $d$-regular $\phi$-expanders from graphs that are $\varepsilon$-far from any $d$-regular $\Omega(\eta\phi^2)$-expanders for any $\eta>0$. The query complexity and running time of the tester are $O(\frac{n^{0.5+\eta}}{\phi^2}(\varepsilon^{-1}\log n)^{O(1)})$, which is almost optimal by a lower bound of $\Omega(\sqrt{n})$ given by Goldreich and Ron~\cite{GR02:testing}. Li, Pan and Peng~\cite{LPP11:conductance} give an expansion tester in the rotation map model with query complexity and running time $\widetilde{O}(\frac{m^{1/2+\eta}}{\phi^2}(\varepsilon^{-1}\log n)^{O(1)})$ for general graphs that matches the best known tester for bounded degree graphs. We remark that when $k=m$, our two-sided tester in the rotation map model can be also guaranteed to test the conductance $\phi(m)$ of $G$ with the same running time and approximation performance. In~\cite{LPP11:conductance}, a product called \textit{non-uniform zig-zag product} was proposed to transfer an arbitrary graph into a bounded degree graph. However, the analysis there is more involved and does not seem to generalize to the $k$-expansion profile for any $k\leq m$ as considered here. Our analysis here is both simple and applicable to the broader case.

The techniques of random walks have also been used to test bipartiteness under different models~\cite{GR99:bipartite,KKR04:bipartite,CMOS11:planar}. In particular, Kaufman et al. extend the bipartiteness tester in bounded degree graphs to general graphs~\cite{KKR04:bipartite} and they also used the idea of replacing high degree vertices by expander graphs. Furthermore, we will also use their techniques for emulating random walks (by performing queries to the oracle of the original graph) and sampling vertices almost uniformly in the transformed graph. However, the transformed graph in~\cite{KKR04:bipartite} may still have large maximum degree (that may be twice the average degree of the original graph), which is not applicable to our case. Ben-Eliezer et al. studied the strength of different query types in the context of property testing in general graphs~\cite{BKKR08:query}. The analysis for the expansion of the replacement product (and the zig-zag product) of two regular graphs are introduced in~\cite{RVW02:zigzag,Rei08:connectivity,RV05:derandomized,RTV06:pseudorandom}.

\subsection{Organization of the paper}
The rest of the paper is organized as follows. In Section~\ref{sec:prelim} we give some basic definitions and introduce the tools for our analysis. Then we introduce the non-uniform replacement product and show its property in Section~\ref{sec:product}. In Section~\ref{sec:tester}, we give all our testers and prove the performance of these testers. Finally, we give a short conclusion in Section~\ref{sec:conclusion}. 

\section{Preliminaries}~\label{sec:prelim}
Let $G=(V,E)$ be an undirected and simple graph with $|V|=n$ and $|E|=m$. Let $\deg_G(v)$ denote the degree of a vertex $v$. As mentioned in the introduction, we consider the \textit{adjacency list} model and the \textit{rotation map} model. In the adjacency list model, the graph is represented by its adjacency list, which is also accessible through an oracle access $\mathcal{O}_G$, and the algorithm is allowed to perform degree and neighbor queries to $\mathcal{O}_G$. In the rotation map model, the graph is represented by its rotation map that for each vertex $u$ and an index $i\leq \deg_G(u)$, in the $(u,i)$th location of the representation the pair $(v,j)$ is stored such that $v$ is the $i$th neighbor of $u$ and $u$ is the $j$th neighbor of $v$. We are given an oracle access $\mathcal{O}_G$ to the rotation map of $G$ and allowed to perform degree queries and (neighbor, index) queries to $\mathcal{O}_G$. We remark that the rotation map model is at least as strong as the adjacency list model. For a graph with maximum degree bounded by $d$, we assume that $d$ is a constant independent of $n$. 

For a vertex subset $S\subseteq V$, let $e_G(S,V\backslash S)$ be the number of edges leaving $S$. Let $\vol_G(S):=\sum_{v\in S}\deg_G(v)$ and $\phi_G(S):=e_G(S,\bar{S})/\vol_G(S)$ be the \textit{volume} and the \textit{conductance} of $S$ in $G$, respectively. Note that $\vol_G(G):=\vol_G(V)=2|E|$. In the following, when it is clear from context, we will omit the subscript $G$. Define the \textit{$k$-expansion profile} of $G$ as $\phi(k):=\min_{S:\vol(S)\leq k}\phi(S)$. In particular, $\phi(m)$ is often referred to the \textit{conductance (or expansion) of $G$} and we let $\phi(G):=\phi(m)$. A graph is called a $\phi$-expander if $\phi(G)\geq \phi$.
\begin{definition}
A graph $G$ is a $(k,\phi)$-expander if $\phi(k)\geq \phi$. Equivalently, $G$ is a $(k,\phi)$-expander if for every $S\subseteq V$ with volume $\vol(S)\leq k$ has conductance $\phi(S)\geq \phi$.
\end{definition}
We have the following definition of graphs that are $\varepsilon$-far from $(k,\phi)$-expanders.
\begin{definition}
A graph $G$ is $\varepsilon$-far from any $(k,\phi)$-expander if one has to modify at least $\varepsilon m$ edges of $G$ to obtain a $(k,\phi)$-expander.
\end{definition}

As mentioned before, we will assume that $m=\Omega(n)$, as otherwise, there exists $n-o(n)$ isolated vertices in $G$, and the graph cannot be a $(k,\phi)$-expander even for constant $k$ and any $\phi>0$. Furthermore, since we will only sample a constant number of vertices (as we do in all our testers), then with high probability, the sampled vertices are all isolated, and in this case, we can safely reject the graph.


We will use bold letters to denote row vectors. For any vector $\p\in\re^V$, let $\p(S):=\sum_{v\in S}\p(v)$ and let $\norm{\p}_1=\sum_{v\in V}|\p(v)|,\norm{\p}_2=\sqrt{\sum_{v\in V}\p(v)^2}$ denote the $l_1,l_2$-norm of $\p$, respectively. Let $\supp(\p)$ be the support of $\p$. Let $\1_S$ be the characteristic vector of $S$, that is, $\1_S(v)=1$ if $v\in S$ and $\1_S(v)=0$ otherwise. Let $\1_v:=\1_{\{v\}}$.

\subsection{Lazy random walks}
We now introduce some tools that will be used in the design and analysis of our algorithms. The following also applies to graphs with possible multiple edges and/or self-loops. First, we define the \textit{lazy random walks} on $G$. In a lazy random walk, if we are currently at vertex $v$, then in the next step, we choose a random neighbor $u$ with probability $1/2\deg(v)$ and move to $u$. With the remaining probability $1/2$, we stay at $v$.

For a given graph $G$, let $A$ denote its adjacency matrix and let $D$ denote the diagonal matrix such that $D_{u,u}=\deg(u)$ for any $u$. Let $I$ denote the identity matrix. Then $W:=(I+D^{-1}A)/2$ is the probability transition matrix of the lazy random walk of $G$. Note that if $\p_0$ is a probability distribution on $V$, then $\p_0W^t$ denotes distribution of the endpoint of a length $t$ lazy random walk with initial distribution $\p_0$. In particular, we let $\p_v^t=\1_vW^t$ be the probability distribution of the endpoint of a walk of length $t$ starting from vertex $v$. Furthermore, we let $\norm{\p_v^t}_2^2$ denote the \textit{collision probability} of such a walk.

For any lazy random walk matrix $W=\frac{I+D^{-1}A}{2}$, 
it is well known that all its eigenvalues are real~(see eg. \cite{OT12:clustering}). Furthermore, if we let $\eta_1(W)\geq\cdots\geq\eta_n(W)$ denote the eigenvalues of $W$, then $0\leq\eta_i(W)\leq 1$ for any $i\leq n$.

\subsection{Spectral characterization of expansion profile}~\label{subsec:spectra}
For a graph $G$, let $\LL:=I-D^{-1/2}AD^{-1/2}$ be the normalized Laplacian matrix of $G$. Let $0=\lambda_1\leq\lambda_2\leq\cdots\leq \lambda_n\leq 2$ be eigenvalues of $\LL$. It is straightforward to verify that $\eta_i=1-\frac{\lambda_i}{2}$ for any $1\leq i\leq n$, where $\eta_i$ is the $i$th largest eigenvalue of the lazy random walk matrix $W$ of $G$. We have the following lemmas relating the expansion profile and the eigenvalues of $\LL$.
\begin{lemma}[Cheeger inequality,~\cite{AM85:lambda,Alo86:eigenvalues,SJ89:approximate}]\label{lem:cheeger}
For every graph $G$, we have $\frac{\lambda_2}{2}\leq\phi(G)\leq\sqrt{2\lambda_2}$.
\end{lemma}

\begin{lemma}[\cite{LOT12:spectral,LRTV12:sparse}]~\label{lem:eigenA}
For every graph $G$, $h\in \mathbb{N}$ and any $\alpha>0$, we have $\phi(\frac{(1+\alpha)2m}{h})\leq O(\frac{1}{\alpha^3}\sqrt{\lambda_h\log h})$.
\end{lemma}
\begin{lemma}[\cite{Ste10:phd,OT12:clustering,OW12:sse}]~\label{lem:eigenB}
For every graph $G$, $h\in \mathbb{N}$ and any $\alpha>0$, we have
$\phi(\frac{2m}{h^{1-\alpha}})\leq O(\sqrt{(\lambda_h/\alpha)\log_h n})$.
\end{lemma}

We remark that in some of references (eg. \cite{LOT12:spectral}), the $k$-expansion profile is defined to be the minimum conductance over all possible subsets of \textit{size} at most $k$, rather than the \textit{volume} measurement as defined here. However, their proofs imply that Lemma~\ref{lem:eigenA} and~\ref{lem:eigenB} also hold for our case.

\subsection{A local algorithm for finding small sparse sets}~\label{subsec:localalgorithm}
We will need the following local algorithm for finding small sparse set to give a one-sided error tester in general graphs as well as to analyze the soundness of our testers. Here, the local algorithm takes as input a vertex $v$ and only explores a small set of the vertices and edges that are ``close'' to $v$, if the volume $k$ of the target set is small. It only needs to perform degree queries and neighbor queries to the oracle of the input graph.
\begin{center}
\begin{tabular}{|p{0.8\textwidth}|}
\hline
\verb|LocalSS|$(G,v,T,\delta)$
\begin{enumerate}
\item Let $\q_0=\1_v$. For each time $0\leq t\leq T$:
\begin{enumerate}
\item Define $\pp_t$ such that $\pp_{t}(u) = \q_{t}(u)$ if $\q_t(u)\geq \delta \deg(v)$ and $\pp_{t}(u) = 0$ if $\q_t(u) < \delta \deg(v)$. Compute $\q_{t+1} := \pp_{t}W$.
\item Let $s_t=|\supp(\pp_t)|$. Order the vertices in $\supp(\pp_t)$ so that $\frac{\pp_{t}(v_1)}{\deg(v_1)} \geq \frac{\pp_{t}(v_2)}{\deg(v_2)} \geq \cdots \geq \frac{\pp_{t}(v_{s_t})}{\deg(v_{s_t})}$.
\item For each $1 \leq i\leq s_t$, let $S_{i,t}$ be the first $i$ vertices in this ordering.
\end{enumerate}
\item Output the subgraph $X$ with the smallest conductance among all the sets $\{S_{i,t}\}_{0\leq t\leq T, 1\leq i\leq s_t}$.
\end{enumerate}\\
\hline
\end{tabular}
\end{center}
The performance of the above algorithm is guaranteed in the following lemma, which follows by combining Proposition 8 in \cite{OT12:clustering} and Theorem 2~\cite{KL12:sparse}. (More specifically, the first part of the lemma is Proposition 8 in \cite{OT12:clustering} and the ``Furthermore'' part of the lemma follows from the proof of Theorem 2~\cite{KL12:sparse}. See also the paragraph ``Independent Work'' in \cite{KL12:sparse}.)

\begin{lemma}~\label{lem:local}
Let $G=(V,E)$ and $t\geq 1$. If $S\subseteq V$ satisfies that $\phi(S)\leq \psi$, then there exists a subset $\widehat{S}\subseteq S$ such that $\vol(\widehat{S})\geq\vol(S)/2$, and for any $v\in \widehat{S}$, we have
\begin{eqnarray*}
\p_v^t(S)\geq c_1(1-\frac{3\psi}{2})^t
\end{eqnarray*}
for some universal constant $c_1>0$. Furthermore, if $\vol(S)\leq k$, then the algorithm \texttt{LocalSS}, with parameters $G, v, T=O(\frac{\zeta\log k}{\psi}), \delta=O(\frac{k^{-1-\zeta}}{T})$ for any $\zeta>0$, will find a set $X$ such that $\vol(X)\leq O(k^{1+\zeta})$ and $\phi(X)\leq O(\sqrt{\psi/\zeta})$. The algorithm can be implemented in time $\widetilde{O}(k^{1+2
\zeta}\psi^{-2})$.
\end{lemma}

\section{Non-uniform replacement product}~\label{sec:product}
In this section, we give the definition of non-uniform replacement product and also show its property, which will be used in our testers for general graphs. Let $G=(V,E)$ be a graph with possible multiple edges or self-loops and with minimum degree $\delta\geq d$. Let $\HH=\{H_u\}_{u\in V}$ be a family of $|V|$ graphs. The graph family $\HH$ is called a \textit{proper $d$-regular graph family} of $G$ if for each $u\in V$, $H_u$ is a $d$-regular graph (with possible parallel edges or self-loops) with vertex set $[\deg_G(u)]:=\{1,...,\deg_G(u)\}$. For any graph $G$ and its proper $d$-regular graph family $\HH$, the \textit{non-uniform replacement product} of $G$ and $\HH$, denoted by $G\R\HH$, is defined as follows.

\begin{enumerate}
\item For each vertex $u$ in $V(G)$, the graph $G\R\HH$ contains a copy of a $H_u$.
\item For any edge $(u,v)\in E(G)$, for each $i\in [\deg_G(u)]$, we specify a unique \textit{but arbitrary} index $j\in [\deg_G(v)]$, and place $d$ parallel edges between the $i$th vertex in $H_u$ and the $j$th vertex in $H_v$.
\end{enumerate}

Now that $G\R\HH$ is a $2d$-regular graph with $2|E|$ vertices. We will use $(u,i)$ to index the vertices in $G\R\HH$. We have the following lemma that formally characterize the intuition that if all the graphs in $\HH$ are expanders, that is, for any $H\in\HH$, $\phi(H)$ is larger than some universal constant, then the expansion profile of $G'$ will not differ by too much from the expansion profile of $G$.

\begin{lemma}~\label{lem:preserve-expansion}
Let $G=(V,E$) be a graph with minimum degree $\delta(G)\geq d$. Let $\HH$ be a proper $d$-regular graph family of $G$, and let $G'=G\R\HH$. We have that
\begin{itemize}
\item If $S\subseteq V(G)$ is a subset with $\phi(S)\leq \phi$ then the set $S':=\{(u,i)\in V(G')|u\in S, 1\leq i\leq \deg_G(u)\}\in V(G')$ satisfies that $|S'|= \vol(S)$ and $\phi_{G'}(S')\leq \phi/2$.
\item If for any set $S\subseteq V(G)$ with $\vol(S)\leq k$, $\phi(S)\geq \phi$ and for any $u$, the conductance of $H_u$ satisfies $\phi(H_u)\geq \delta$, then for any set $S'\subseteq V(G')$ with $|S'|\leq \Theta(k)$, $\phi_{G'}(S')=\Omega(\delta\phi^2)$.
\end{itemize}
\end{lemma}
\begin{proof}
The first part of the lemma is straightforward. By the definition of $S'$, $S'$ is the set consisting of all vertices in $H_u$ for any $u\in H$. Thus, $|S'|=\sum_{u\in S}\deg_G(u)=\vol(S)$. Furthermore, since $\phi(S)\leq \phi$, then $e(S,V\backslash S)\leq \phi\vol(S)$. By our construction of $G'$, the number of edges between $S'$ and $V(G')\backslash S'$ is
\begin{displaymath}
d \cdot e(S,V\backslash S)
    \leq
d \cdot \phi \vol(S)
    =
\phi \vol_{G'}(S')/2,
\end{displaymath}
which gives that $\phi_{G'}(S')\leq \phi/2$.

The second part of the lemma follows by the same arguments given in the proof of Theorem 1.3 in~\cite{ASS08:expander}. Actually our case is even simpler, since we only need to consider all sets $S'$ with size at most $\Theta(k)$ rather than $m/2$. We omit the details here.
\end{proof}

When the rotation map of the graph $G$ is explicitly given, we define the \textit{non-uniform replacement product with rotation map} of $G$ and $\HH$, denoted as $G^{(r)}\R\HH$, as follows.
\begin{enumerate}
\item For each vertex $u$ in $V(G)$, the graph $G^{(r)}\R\HH$ contains a copy of a $H_u$.
\item For any edge $(u,v)\in E(G)$ such that $v$ is the $i$th neighbor of $u$ and $u$ is the $j$th neighbor of $v$, we place $d$ parallel edges between the $i$th vertex in $H_u$ and the $j$th vertex in $H_v$.
\end{enumerate}

Note that the above replacement product with rotation map is a special case of the (general) replacement product defined before. Thus, it not only satisfies the combinatorial property of expansion profile given in Lemma~\ref{lem:preserve-expansion}, bust also satisfies the following nice spectral properties.

\begin{lemma}~\label{lem:separation}
Let $G=(V,E$) be a graph with minimum degree $\delta(G)\geq d$. Let $\HH$ be a proper $d$-regular graph family of $G$, and let $G'=G^{(r)}\R\HH$ be the replacement product with rotation map of $G$ and $\HH$. We have that
\begin{itemize}
\item $G'$ satisfies the two properties in Lemma~\ref{lem:preserve-expansion}.
\item If for any set $S\subseteq V(G)$ with $\vol(S)\leq k$, $\phi(S)\geq \phi$ and for any $u$, $\eta_2(W_{H_u})\leq 1-\delta$ for some $\delta>0$, then for any $\alpha>0$, \begin{eqnarray*}
    \eta_{\frac{(1+\alpha)2m}{k}}(W_{G'}) &\leq&  1-\Omega(\delta^2\alpha^6\phi^2(\log\frac{2m}{k})^{-1}),\\ \eta_{(2m/k)^{1+\alpha}}(W_{G'}) &\leq& 1-\Omega(\alpha\delta^2\phi^2\log_n(2m/k)).
    \end{eqnarray*}
     Furthermore, when $k=m$, we have
     \begin{eqnarray*}
     \eta_{2}(W_{G'})\leq 1-\Omega(\delta^2\eta^2).
     \end{eqnarray*}
\end{itemize}
\end{lemma}

We defer the proof of the above lemma in Section~\ref{subsec:proof-separation} and now we use it to prove Theorem~\ref{thm:spectra}. 
\begin{proof}[Proof of Theorem~\ref{thm:spectra}]
For any graph $G=(V,E)$, we first turn it into a graph $G_{\geq 8}$ with minimum degree $8$ by adding an appropriate number of self-loops to vertices with degree smaller than $8$. Note that this only changes the conductance of a set by a factor of $8$. Now we let $\HH$ be a proper $8$-regular graph family for $G_{\geq 8}$ such that for any $u\in V$, $H_u$ is a Margulis expander with $\deg_{G_{\geq 8}}(u)$ vertices~\cite{Mar73:construction,GG81:explicit}. Therefore, each $H_u$ is an expander such that $\phi(H_u)$ and $1-\eta_2(W_{H_u})$ are larger than some universal constants. Then we let $G'=G_{\geq 8}^{(r)}\R\HH$, $d=8$ and specify $\delta$ to be a constant in Lemma~\ref{lem:separation}. By definition, $G'$ is a $16$-regular graph. Finally, the theorem follows by Lemma~\ref{lem:separation} and the fact that $\eta_i=1-\frac{\lambda_i}{2}$.
\end{proof}

\subsection{Proof of Lemma~\ref{lem:separation}}~\label{subsec:proof-separation}
Now we turn to prove Lemma~\ref{lem:separation}. Note that we only need to prove the second part of the lemma. We first give a useful lemma for the proof of Lemma~\ref{lem:separation}. Let $J_n$ denote the $n\times n$ matrix with all elements equal to $\frac{1}{n}$. Recall that $\eta_i(W)$ is the $i$th  largest eigenvalue of matrix $W$.
\begin{lemma}~\label{lem:matrixdecomp}
Let $H$ be a $d$-regular graph on $n$ vertices and let $W$ be its lazy random walk matrix. If $\eta_2(W)\leq 1-\delta$ for some $0\leq \delta<1$, then \begin{displaymath}
W = \delta J_n + (1 - \delta) B,
\end{displaymath}
where $\eta_1(B)\leq 1$.
\end{lemma}
\begin{proof}
Define $B:=\frac{W_H-\delta J_n}{1-\delta}$. Since $H$ is $d$-regular, $W$ is symmetric and thus we can find orthonormal eigenvectors $\vv_1,\cdots,\vv_n$ of $W$ with corresponding eigenvalues $\eta_1(W),\cdots,\eta_n(W)$ such that $\{\vv_i\}_{i=1}^n$ form an orthonormal basis of $\re^V$. Then by the spectral decomposition theorem, $W=\sum_{i=1}^n\eta_i(W)\vv_i^T\vv_i$. Noting that $\eta_1(W)=1$ and $\vv_1=(\frac{1}{\sqrt{n}},\cdots,\frac{1}{\sqrt{n}})$, we have that $J_n=\vv_1^T \vv_1$ and thus
\begin{displaymath}
B = \frac{(1 - \delta)\vv_1^T \vv_1 + \sum_{i=2}^n \eta_i(W) \vv_i^T \vv_i}{1 - \delta}.
\end{displaymath}
Now for any nonzero vector $\x$, we can write it as $\x=\sum_{i=1}^n\alpha_i\vv_i$, and thus
\begin{displaymath}
\frac{\x B\x'}{\x\x'} = \frac{\alpha_1^2 + \sum_{i=2}^n \alpha_i^2 \frac{\eta_i(W)}{1 - \delta}}{\sum_{i=1}^n \alpha_i^2}\leq\frac{\sum_{i=1}^n \alpha_i^2}{\sum_{i=1}^n \alpha_i^2}=1,
\end{displaymath}
where the last inequality follows from the fact that for any $i\geq 2$, $\eta_i\leq \eta_2$ and the precondition that $\eta_2(W)\leq 1-\delta$. Therefore,
\begin{displaymath}
\eta_1(B) = \max_{\x\in \re^V,\x \neq \vec{0}} \frac{\x B\x'}{\x\x'}\leq 1.
\end{displaymath}
\end{proof}

The following observation is also helpful to the proof of Lemma~\ref{lem:separation}. Let $R_G$ be the permutation matrix corresponding the rotation map of $G$. That is, $R_G$ is an $\vol(G)\times \vol(G)$ matrix such that for each row indexed $(u,i)$, only in the column indexed $(v,j)$ the entry is $1$, and in any other column, the entry is $0$, where $v$ is the $i$th neighbor of $v$ and $u$ is the $j$th neighbor of $i$. For each $u\in V$, let $W_{H_u}$ be the lazy random walk of $H_u$ and let $W_\HH$ be the block diagonal matrix with each block $W_{H_u}$. Note that by our construction, the lazy random walk matrix $W_{G'}$ of $G'$ satisfies that \begin{displaymath}
W_{G'} = \frac{1}{2}(\frac{I+R_G}{2}+W_\HH).
\end{displaymath}
Now we are ready to prove Lemma~\ref{lem:separation}.

\begin{proof}[Proof of Lemma~\ref{lem:separation}]

Now we prove the second part of the lemma. For each $u\in V$, let $W_{H_u}$ denote the lazy random walk matrix of $H_u$. Then by Lemma~\ref{lem:matrixdecomp} and the assumption that $\eta_2(W_{H_u})\leq 1-\delta$ for every $u$, we have $W_{H_u} = \delta J_{\deg_G(u)} + (1 - \delta) B_u$,
where $\eta_1(B_u)\leq 1$. Let $B_\HH$ (resp., $J_\HH$) be the block diagonal matrix with each block $B_{u}$ (resp., $J_{\deg_G(u)}$). Therefore $W_\HH = \delta J_\HH+(1-\delta)B_\HH$, and $
W_{G'} = \frac{1}{2}(\frac{I+R_G}{2}+W_\HH) = \frac{1}{2}(\frac{I+R_G}{2}+\delta J_\HH+(1-\delta)B_\HH)$. The latter gives that
\begin{displaymath}
W_{G'}^3=\frac{1}{8}((I+R_G)/2+\delta J_\HH+(1-\delta)B_\HH)^3.
\end{displaymath}
Then we expand all terms to get that \begin{displaymath}
W_{G'}^3=(1-\frac{\delta^2}{8})B+\frac{\delta^2}{8}\cdot\frac{1}{2}J_\HH (I+R_G)J_\HH,
\end{displaymath}
where $B$ is some matrix with $\eta_1(B)\leq 1$. Let $P:=\frac{1}{2}J_\HH (I+R_G)J_\HH$. By Weyl's inequality~\cite{Tao12:randommatrix}, we have that for any $j\leq 2m$,
\begin{eqnarray}
\eta_j(W_{G'}^3) \leq  (1 - \frac{\delta^2}{8})\eta_1(B) + \frac{\delta^2}{8}\eta_j(P) \leq 1 - \frac{\delta^2}{8} + \frac{\delta^2}{8}\eta_j(P). ~\label{eqn:eigen-bound}
\end{eqnarray}

Now we bound the eigenvalues of $P$. We need the following two claims. 
\begin{claim}
For any $(u,i),(v,j)\in V(G')$, $P_{(u,i),(v,j)}=\frac{(D_G^{-1}A_G)(u,v)+I(u,v)}{2\deg_G(v)}=\frac{W_G(u,v)}{\deg_G(v)}$.
\end{claim}
\begin{proof}
Since $P=J_\HH R_GJ_\HH$, then $P$ can be seen as the random walk matrix on $V(G')$ that does the following from $(u,i)$: first chooses a random number $k_1$ from set $[\deg_G(u)]:=\{1,\cdots,\deg_G(u)\}$, and then 1) with half probability it stays at $(u,k_1)$, and then chooses a random number $k_2$ from $[\deg_G(u)]$ and goes to $(u,k_2)$. 2) with the remaining half probability, it goes to $(v,j)$, where $v$ is the $k_1$th neighbor of $u$ and $u$ is the $j$th neighbor of $v$, and then choose a random number $k_2$ from $[\deg_G(v)]$ and goes to $(v,k_2)$. This process is equivalent to first perform the lazy random walk from $u$ to $v$, and then choose a random number $k_2$ from $[\deg_G(v)]$ and output $(v,k_2)$. Such an equivalence is exactly characterized by the statement of the lemma. This completes the proof.
\end{proof}
\begin{claim}There is a one-to-one correspondence between the nonzero eigenvalues of $P$ and the nonzero eigenvalues of $M_G:=\frac{I+A_GD_G^{-1}}{2}$.
\end{claim}
\begin{proof}
On one hand, let $\x\in\re^{V(G')}$ be an eigenvector of $P$ with eigenvalue $\eta\neq 0$. Then for any $(v,j)$,
\begin{eqnarray*}
\eta\x_{(v,j)} = \sum_{(u,i)}\x_{(u,i)}P_{(u,i),(v,j)} = \sum_{(u,i)}\x_{(u,i)}\frac{W_G(u,v)}{\deg_G(v)},
\end{eqnarray*}
which is independent of $j$. Since $\eta\neq 0$, this means that for any $v$ and $1\leq j_1,j_2\leq \deg_G(v)$, $\x_{(v,j_1)}=\x_{(v,j_2)}$. Furthermore, if we let $\y_v=\x_{(v,j)}$, then by the above calculation, for any $v\in V(G)$,
\begin{eqnarray*}
(\y M_G)_v=\sum_u\deg_G(u)\y_u\frac{W_G(u,v)}{\deg_G(v)}=\sum_u\y_uM_G(u,v)=\eta \y_v,
\end{eqnarray*}
which gives that $\y$ is the eigenvector of $M_G$ with eigenvalue $\eta$.

On the other hand, let $\y\in\re^{V(G)}$ be an eigenvector of $M_G$ with eigenvalue $\eta$. Then we define for each $u$ and $1\leq i\leq\deg_G(u)$, $\x_{(u,i)}=\y_u$. Then for any $(v,j)\in V(G')$,
\begin{eqnarray*}
(\x P)_{(v,j)}=\sum_{(u,i)}\x_{(u,i)}P_{(u,i),(v,j)} = \sum_{(u,i)}\x_{(u,i)}\frac{W_G(u,v)}{\deg_G(v)} &=& \sum_{(u,i)}\x_{(u,i)}\frac{M_G(u,v)}{\deg_G(u)}\\
& =& \sum_u\y_uM_G(u,v)\\
&=&\eta\y_v \\
&=&\eta\x_{(v,j)},
\end{eqnarray*} which means that $\x$ is an eigenvector of $P$ corresponding to eigenvalue $\eta$.
\end{proof}
Now note that the eigenvalues of $M_G$ are the same as the eigenvalues of $W_G$ since $W_G=D_G^{-1}M_GD_G^{1}$, and both the eigenvalues of $P$ and $M_G$ are non-negative. These two facts combined the above two claims implies that for any $i\leq n$, $\eta_i(P)\leq \eta_i(W_G)=1-\frac{\lambda_i}{2}$, where $\lambda_i$ is the $i$th smallest eigenvalue of the Laplacian matrix $\LL$ of $G$. Finally, by the fact that $\phi_G(k)\geq \phi$, Lemma~\ref{lem:eigenA} and~Lemma \ref{lem:eigenB}, we get that for any $\alpha>0$,
\begin{eqnarray*}
\eta_{\frac{(1+\alpha)2m}{k}}(P) &\leq& 1-\Omega(\alpha^6\phi^2(\log\frac{2m}{k})^{-1}),\\ \eta_{(2m/k)^{1+\alpha}}(P) &\leq & 1-\Omega(\alpha\phi^2\log_n\frac{2m}{k}).
\end{eqnarray*}
By inequality~(\ref{eqn:eigen-bound}), this further gives that
\begin{eqnarray*}
\eta_{\frac{(1+\alpha)2m}{k}}(W_{G'}^3)&\leq &1 - \Omega(\delta^2\alpha^6\phi^2(\log\frac{2m}{k})^{-1}),\\
\eta_{(2m/k)^{1+\alpha}}(W_{G'}^3)&\leq& 1-\Omega(\alpha\delta^2\phi^2\log_n\frac{2m}{k}),
\end{eqnarray*}
and the first two inequalities in the statement of the lemma then follows by noting that $\eta_j(W_{G'}^3)=(\eta_j(W_{G'}))^3$.

The ``Furthermore'' part of the lemma follows from the above analysis and the Cheeger inequality given in Lemma~\ref{lem:cheeger}.
\end{proof}

\section{Testers for small set expansion}~\label{sec:tester}
In this section, we give all our testing algorithms for small set expansion. We first show a property of graphs that are far from small set expander in Section~\ref{subsec:property-far}, which will be useful for all our testers. Then in Section~\ref{subsec:bounded-tester}, we give a two-sided error tester in bounded degree model, which illustrates basic ideas underlying our algorithms. Finally, we give testers in adjacency list model and in the rotation map for general graphs in Section~\ref{subsec:tester-list-model},~\ref{subsec:tester-map-model}, respectively.

\subsection{A property of graphs that are far from small set expander}~\label{subsec:property-far}
The following lemma shows that if a general graph $G$ is far from $(k,\phi)$-expander, then there exist disjoint subsets such that each of them is of small size and small conductance, and the total volume of these sets are large. This lemma will be useful for the analysis of all the testers. 
\begin{lemma}~\label{lem:partition}
Let $c_2$ be some constant and let $\phi^*\leq\frac{1}{20c_2}$. If a graph $G$ is $\varepsilon$-far from $(k^*,\phi^*)$-expander, then there exist disjoint subsets $S_1,\cdots,S_q\subseteq V$ such that $\vol(S_1\cup\cdots\cup S_q)\geq \frac{\varepsilon m}{15}$, and for each $i\leq q$, $\vol(S_i)\leq 2k^*$, $\phi(S_i)<11c_2\phi^*$.
\end{lemma}
To prove Lemma~\ref{lem:partition}, we first introduce a useful lemma that is implied in the proof of Lemma 8, 9 and 10 in~\cite{LPP11:conductance}, which in turn generalize a corresponding result for bounded degree graphs in \cite{KS11:tester} and ~\cite{CS10:expansion}.
\begin{lemma}[\cite{LPP11:conductance}]~\label{lem:lpp}
Let $G=(V,E)$ and let $c_2$ be some constant. If there exists a set $A\subseteq V$ such that $\vol(A)\leq \frac{\varepsilon \vol(G)}{20}$, and the subgraph $G[V\backslash A]$ is a $\phi^*$-expander, then there exists an algorithm that modifies at most $\varepsilon m$ edges to get a $c_2\phi^*$-expander $G'=(V,E')$ such that for each $v\in V$, $\deg_{G'}(v)\leq \deg_G(v)$.
\end{lemma}

The following result is a direct corollary of the above lemma. That is, we can use the same proof and modification algorithm of Lemma~\ref{lem:lpp} to show that if $G[V\backslash A]$ is a $(k^*,\phi^*)$-expander for some set $A$ with small volume, then $G$ is not $\varepsilon$-far from $(k^*,c_2\phi^*)$-expander. Actually, the proof in~\cite{LPP11:conductance} studies the expansion of all possible sets of volume at most $\frac{\vol(G)}{2}$, here we only need to consider sets of volume at most $k^*\leq \frac{\vol(G)}{2}$.

\begin{corollary}~\label{cor:patchup}
If there is a set $A\subseteq V$ with $\vol(A)\leq \frac{\varepsilon m}{10}$ such that $G[V\backslash A]$ is a $(k^*,\phi^*)$-expander, then $G$ is not $\varepsilon$-far from a $(k^*,c_2\phi^*)$-expander. Furthermore, if the maximum degree of $G$ is bounded by some constant $d$, then $G$ is not $\varepsilon$-far from to a $(k^*,c_2\phi^*)$-expander with maximum degree at most $d$.
\end{corollary}

For a set $S\subset V$ and $T\subseteq S$, we use $\vol_S(T)$ and $\phi_S(T)$ to denote the volume and conductance of $T$ measured in the induced subgraph $G[S]$. If $S=V$, we drop the subscript of $\vol_S(T),\phi_S(T)$. We let $e(S)$ denote the number of edges in $S$.
\begin{proof}[Proof of Lemma~\ref{lem:partition}]
We perform the following algorithm on $G$. Let $A_0$ be the empty set and let $V_0:=V$. For each $i\geq 1$, if $\vol(\cup_{j\leq i-1}A_j)\leq \frac{\varepsilon m}{10}$, then we apply Corollary~\ref{cor:patchup} with $A=\cup_{j\leq i-1}A_j$ to find a subset $A_i\subseteq V_{i-1}$ such that $\vol_{G[V_{i-1}]}(A_i)\leq k^*$ and $\phi_{G[V_{i-1}]}(A_i)<c_2\phi^*$, then we remove $A_i$ from $V_{i-1}$ and let $V_{i}:=V_{i-1}\backslash A_i$. By Corollary~\ref{cor:patchup}, we can repeat this process until at some time $s$, $\vol(A_1\cup\cdots\cup A_s)\geq\frac{\varepsilon m}{10}$. 
We let $P=A_1\cup\cdots\cup A_s$.

Note that
\begin{eqnarray*}
\sum_{i=1}^s e(A_i,V\backslash A_i) \leq 2\sum_{i=1}^s e(A_i,V_{i-1}\backslash A_i)
&\leq& 2 \sum_{i=1}^s \vol_{V_{i-1}}(V_i) c_2 \phi^* \\
&\leq& 2 c_2 \vol(P) \phi^*.
\end{eqnarray*}

Now we call an index $i$ \textit{bad}, if $\vol_{V_{i-1}}(A_i)<(1-10c_2\phi^*)\vol(A_i)$, and \textit{good} otherwise, for each $1\leq i\leq s$. Note that for a bad index $i$ and the corresponding set $A_i$, since $2e(A_i)\leq \vol_{V_{i-1}}(A_i)$, we have
\begin{eqnarray*}
e(A_i,V\backslash A_i)=\vol(A_i)-2e(A_i)\geq\vol(A_i)-\vol_{V_{i-1}}(A_i)> 10c_2\phi^*\vol(A_i),
\end{eqnarray*}
where the last inequality follows by our definition of bad indices. Therefore,
\begin{displaymath}
\sum_{i: bad}\vol(A_i)< \frac{1}{10c_2\phi^*}\sum_{i: bad}e(A_i,V\backslash A_i)< \frac15\vol(P).
\end{displaymath}
This means that
\begin{displaymath}
\sum_{i: good}\vol(A_i)\geq (1-\frac15)\vol(P)\geq \frac{\varepsilon m}{15},
\end{displaymath}
and for each good $i$, by our assumption that $\phi^*\leq\frac{1}{20c_2}$,
\begin{eqnarray*}
\vol(A_i)\leq \frac{1}{1-10c_2\phi^*}\vol_{V_{i-1}}(A_i)\leq 2k^*.
\end{eqnarray*}
Furthermore, by definition, \begin{eqnarray*}
\phi_{V_{i-1}}(A_i)=\frac{e(A_i,V_{i-1}\backslash A_i)}{\vol_{V_{i-1}}(A_i)}=\frac{\vol_{V_{i-1}}(A_i)-2e(A_i)}{\vol_{V_{i-1}}(A_i)}\leq c_2\phi^*,
\end{eqnarray*}
we have that
\begin{eqnarray*}
2e(A_i)\geq (1-c_2\phi^*)\vol_{V_{i-1}}(A_i)\geq (1-c_2\phi^*)(1-10c_2\phi^*)\vol(A_i)\geq (1-11c_2\phi^*)\vol(A_i),
\end{eqnarray*}
which gives that \begin{eqnarray*}
\phi(A_i)=\frac{\vol(A_i)-2e(A_i)}{\vol(A_i)}\leq 11c_2\phi^*.
\end{eqnarray*}
The lemma                                               follows by specifying $S_j$ to be sets $A_i$ with good indices $i$.
\end{proof}
\subsection{A tester for bounded degree graphs}~\label{subsec:bounded-tester}
Now we give a two-sided error tester for bounded degree graphs. This tester is very intuitive and simple: we sample a small number of vertices, and for each sampled vertex $v$, we perform independently a number of random walks from $v$ and calculate the number of collisions $Z_v$ between the endpoints of these random walks. We accept the graph if and only if $Z_v$ is small for every sampled vertex $v$. We remark that this idea originates from the tester for expansion for bounded degree graphs~\cite{GR00:expansion,CS10:expansion,KS11:tester,NS10:expansion}. The main difference between our small set expansion tester and the previous expansion testers is the choice of parameters.

Given a $d$-bounded degree graph $G$, we define the following $d$-regularized random walk on $G$: at each vertex $v$, with probability $\deg_G(v)/2d$, we jump to a randomly chosen neighbor of $v$, and with the remaining probability $1-\frac{\deg_G(v)}{2d}$, we stay at $v$. This random walk is equivalent to the lazy random walk on the virtually constructed $d$-regular graph $G_\reg$ that is obtained by adding an appropriate number of self-loops on each vertex in $G$. Note that to perform such a random walk, we only need to perform neighbor queries to the oracle of $G$. Our tester for bounded degree graphs is as follows.

\begin{center}
\begin{tabular}{|p{0.8\textwidth}|}
\hline
\texttt{SSETester2-Bound}$(G,s,r,\ell,\sigma)$
\begin{enumerate}
\item Repeat $s$ times:
\begin{enumerate}
\item Select a vertex $v$ uniformly at random from $V$.
\item Perform $r$ independent $d$-regularized random walks of length $\ell$ starting from $v$.
\item Let $Z_v$ be the number of pairwise collisions among the endpoints of these $r$ random walks.
\item If $Z_v>\sigma$ then abort and output \textbf{reject}.
\end{enumerate}

\item Output \textbf{accept}.
\end{enumerate}\\
\hline
\end{tabular}
\end{center}

We can show that by choosing appropriate parameters, the above algorithm is a property tester for small set expansion for bounded degree graphs. We have the following theorem. 

\begin{theorem}~\label{thm:bounded}
Given neighbor query access to a $d$-bound-degree graph $G$, a volume bound $k$, a distance parameter $\varepsilon$ and a conductance bound $\phi$, then the algorithm \texttt{SSETester2-Bound} with parameters $s=\Theta(1/\varepsilon)$, $r=\Theta(\sqrt{n}/\varepsilon)$, $\ell=\Theta(\frac{(\ln k)\cdot\log(2nd/k)}{\phi^2})$ and $\sigma=\binom{r}{2}\frac{60}{k\varepsilon}$, accepts any $(k,\phi)$-expander graph $G$ with degree bounded by $d$ and rejects any graph that is $\varepsilon$-far from $(k^*,\phi^*)$-expander with degree bounded by $d$, where $k^*=\Theta(k\varepsilon/d),\phi^*=\Theta(\frac{\phi^2}{(\ln k)\cdot\log(2nd/k)})$, with probability at least $2/3$. The query complexity and running time are $\widetilde{O}(\sqrt{n}\phi^{-2}\varepsilon^{-2})$.
\end{theorem}

To prove the above theorem, we need the following properties of lazy random walks. Let $H=(V,E)$ be a $d$-regular graph with possible self-loops. (It will be helpful to think of $H$ as the regularized version $G_\reg$ of the input graph $G$). Let $t,r\geq 1$. Let $\p_v^t$ denote the distribution of the endpoints of a lazy random walk of length $t$ from $v$ in $H$. Consider $r$ independent samples from $\p_v^t$. For any vertex $v\in V$, let $Z_v$ denote the number of pairwise collisions among these samples. We have the following lemma that follows from the first paragraph of the proof of Lemma 4.2 in~\cite{CS10:expansion} (which in turn follows from Lemma 1 in~\cite{GR00:expansion}) by setting $\varepsilon=1/2$ there.
\begin{lemma}~\label{lem:collision}
If $r\geq 16\sqrt{|V|}$, then with probability at least $1-\frac{16\sqrt{|V|}}{r}$, $\frac12\binom{r}{2}\norm{\p_v^t}_2^2\leq Z_v\leq\frac32\binom{r}{2}\norm{\p_v^t}_2^2$.
\end{lemma}

Now we are ready to prove Theorem~\ref{thm:bounded}. In the following, we let $c$ be a sufficiently large constant.
\begin{proof}[Proof of Theorem~\ref{thm:bounded}]
We set $s=\Theta(1/\varepsilon)$, $r=\Theta(s\sqrt{n})$, $\ell=\Theta(\frac{(\ln k)\cdot\min\{\log(2nd/k),\log n\}}{\phi^2})$ and $\sigma=\binom{r}{2}\frac{6dc}{k\varepsilon}$ in the algorithm \texttt{SSETester2-Bound}. Let $k^*=\Theta(k\varepsilon/d)$ and $\phi^*=\Theta(\frac{\phi^2}{(\ln k)\cdot\min\{\log(2nd/k),\log n\}})$. Note that $G$ is a $d$-bounded degree graph, where $d$ is constant. Let $H:=G_\reg$ denote the $d$-regularized version of $G$, and the number of edges in $H$ is $m_H:=nd/2$. Note that the $d$-regularized random walk is equivalent to the lazy random walk on $H$.

\begin{lemma}[Completeness]~\label{lem:completeness-bounded}
If $G$ is a $(k,\phi)$-expander, then \texttt{SSETester2-Bound} accepts $G$ with probability at least $2/3$.
\end{lemma}
\begin{proof}
Since $G$ is a $(k,\phi)$-expander, then it is straightforward to see that $H$ is also a $(k,\phi)$-expander. Let $\lambda_i$ (resp. $\eta_i$) be the $i$th smallest (resp. largest) eigenvalue of the Laplacian (resp. lazy random walk matrix) of $H$. By applying $\alpha=1$ in Lemma~\ref{lem:eigenA}, we have $\lambda_{4m_H/k}\geq \Omega(\phi^2\frac{1}{\log(4m_H/k)})$. By applying $\alpha=\frac{1}{\log(4m_H/k)}$ in Lemma \ref{lem:eigenB}, we have $\lambda_{4m_H/k}\geq \Omega(\phi^2\frac{1}{\log n})$. Thus, $\lambda_{4m_H/k} \geq \kappa$ for 
$\kappa:=\Omega(\frac{\phi^2}{\min\{\log(4m_H/k),\log n\}})$.
This further gives that $\eta_{4m_H/k} = 1 - \frac{\lambda_{4m_H/k}}{2} \leq 1-\kappa/2$.

Note that for any $t\geq 1$, the trace of matrix $W_{H}^{2t}$, denoted $\Tr(W_{H}^{2t})$, satisfies that $\Tr(W_{H}^{2t})=\sum_{i=1}^{n}\eta_i^{2t}\leq \sum_{i=1}^{4m_H/k}\eta_i^{2t}+n\cdot\eta_{4m_H/k}^{2t}\leq {4m_H/k}+n(1-\kappa/2)^{2t}$. By setting $t=\ell$, we have that $\Tr(W_{H}^{2\ell})\leq \frac{8m_H}{k}$.

On the other hand, $\Tr(W_{H}^{2\ell})=\sum_{v\in V(H)}\norm{\1_{v}W_{H}^t}_2^2$. Thus, the average value of $\norm{\1_{v}W_{H}^t}_2^2$ over all $n$ possible vertices $v$ is at most $\frac{8m_H}{nk}=\frac{4d}{k}$. Furthermore, if we let $U:=\{v|\norm{\1_{v}W_{H}^t}_2^2<\frac{4dc}{k\varepsilon}\}$. Then by Markov's inequality, $|U|\geq (1-\varepsilon/c)n$. Therefore, the probability that all the sampled vertices are in $U$ is at least $(1-\varepsilon/c)^s\geq 5/6$, since $s=\Theta(1/\varepsilon)$ and $c$ is a sufficiently large constant.

Now we assume that all the sampled vertices are in $U$. By Lemma~\ref{lem:collision} and the definition of $U$, we know that for each sampled vertex $v$, $Z_{v}\leq \frac{3}{2}\binom{r}{2}\frac{4dc}{k\varepsilon}=\sigma$ holds with probability at least $1-\frac{16\sqrt{n}}{r}\geq 1-\frac{1}{10s}$, where the last inequality follows from our choice that $r=\Theta(\sqrt{n}s)$. Then with probability at least $1-\frac{1}{10s}\cdot s\geq \frac56$, for all sampled vertices $v$, $Z_{v}\leq \sigma$, and thus the tester will accept $G$.

Overall, the probability that the tester will accept $G$ is $\frac{5}{6}\cdot\frac56\geq \frac23$.
\end{proof}

\begin{lemma}[Soundness]
If $G$ is $\varepsilon$-far from $(k^*,\phi^*)$-expander, then \texttt{SSETester2-Bound} rejects $G$ with probability at least $2/3$.
\end{lemma}
\begin{proof}
First, we note that if $G$ is $\varepsilon$-far from any $(k^*,\phi^*)$-expander, then $H$ is $\varepsilon$-far from any $(dk^*,\phi^*)$-expander. We apply Lemma~\ref{lem:partition} with $G=H$, and let $S_1,\cdots,S_q$ be the sets with properties guaranteed in the statement of the lemma. That is, $S_1,\cdots, S_q$ are disjoint, and $\vol_H(S_i)\leq 2dk^*, \phi_H(S_i)\leq 11c_2\phi^*$ and $\vol_H(S_1\cup\cdots\cup S_q)\geq \frac{\varepsilon m}{15}$. Now we apply Lemma~\ref{lem:local} with $G=H, S=S_i$ and $t=\ell, \psi=11c_2\phi^*$ to find $\widehat{S_i}\subseteq S_i$ such that $|\widehat{S_i}|\geq \frac{1}{2}|S_i|$ and for each vertex $v\in \widehat{S_i}$, if we let $\p=\1_{v}W_{H}^\ell$, then $\norm{\p}_2^2\geq \sum_{u\in S_i}\p_{u}^2\geq\frac{\sum_{u\in S_i}\p_{u}}{|S_i|}\geq \frac{c_1}{2dk^*}(1-\frac{33c_2\phi^*}{2})^\ell\geq \frac{13dc}{k\varepsilon}$, where we used the Cauchy-Schwarz inequality, the fact that $|S_i|\leq \vol_H(S_i) \leq 2dk^*$ and our choice of parameters.

Thus, if we have sampled some vertex $v\in \widehat{S_i}$ for some $i\leq q$, then the collision probability of the corresponding random walk will be at least $\frac{13dc}{k\varepsilon}$. Then by Lemma~\ref{lem:collision}, with probability at least $1-\frac{16\sqrt{n}}{r}>5/6$, $Z_{v}\geq\frac{1}{2}\binom{r}{2}\frac{13dc}{k\varepsilon}>\sigma$, where the last inequality follows by our choice of $\sigma$, and then the tester will reject the graph $G$.

Now 
note that since $\vol_H(S_1\cup\cdots\cup S_q)\geq \frac{\varepsilon m}{15}=\frac{\varepsilon nd}{30}$, each vertex in $H$ has degree $d$, then $|S_1\cup\cdots\cup S_q|\geq \frac{\varepsilon n}{30}$. Since we sampled $\Theta(1/\varepsilon)$ vertices, each with probability $1/n$, we can guarantee that with probability at least $5/6$, the algorithm will sample out a vertex from $\widehat{S_1}\cup\cdots\cup \widehat{S_q}$.

Therefore, the overall probability that the tester will reject $G$ is at least $\frac{5}{6}\cdot\frac56>\frac23$.
\end{proof}
Finally, it is straightforward to see that the query complexity and the running time of our algorithm is $O(r\cdot s\cdot \ell) = O(\frac{\sqrt{n}(\ln k)\cdot\min\{\log(2nd/k),\log n\})}{\varepsilon ^ 2 \phi ^ 2})$. Theorem~\ref{thm:bounded} then follows by noting that $\Theta(\min\{\log(2nd/k),\log n\})=\Theta(\log(2nd/k))$.
\end{proof}


\subsection{Testers in the adjacency list model for general graphs}~\label{subsec:tester-list-model}
In this section, we give testers for small set expansion for general graphs in the adjacency list model.
\subsubsection{A two-sided error tester}~\label{subsec:two-sided-list}
To give a two-sided error tester for general graphs, we first note that the tester for bounded degree graphs given in Section~\ref{subsec:bounded-tester} does not apply to general graph, which may have an arbitrary large degree. For example, in a star graph the collision probability of a lazy random walk will be very large on the ``central'' vertex, however, the conductance of star graph is large and it is thus a small set expander. This implies that we cannot directly apply our tester for bounded degree graphs to general graphs.

In the following, we show that we can use the non-uniform replacement product (without rotation map) defined in Section~\ref{sec:product} to first turn our input graph $G$ into a bounded degree graph $G'$, and then we perform independent random walks on the newly transformed graphs $G'$ to determine whether to accept or reject the input graph $G$. We should keep in mind that we are only given degree and neighbor query access to $G$ rather than $G'$.

We first define $G'$. To do so, we first specify a proper $d$-regular graph family $\HH$ for $G$. We will let $d=8$, and first turn $G$ into a graph $G_{\geq 8}$ with minimum degree $8$ by adding an appropriate number of self-loops on vertices with degree smaller than $8$. Note that this modification only changes the conductance of a set by a factor of $8$. Now we let $\HH$ be the graph family that for any $u\in G$, $H_u$ is a Margulis expander with $\deg_{G_{\geq 8}}(u)$ vertices. We stress that such expanders are explicitly constructible~\cite{Mar73:construction,GG81:explicit}. Furthermore, given any vertex $i\in H_u$, we can determine the neighborhood of $i$ in constant time. Now we define $G'=G_{\geq 8}\R\HH$.

By definition of $G'$, we can specify a vertex $(u,i)$ to connect to a vertex in $\cup_{v:(v,u)\in E}H_v$ in an \textit{arbitrary} manner. This important property allows us to construct $G'$ when we go along and emulate random walks in $G'$ very efficiently by performing degree and neighbor queries to $G$. We stress here that if the non-uniform replacement product with rotation map of $G$ is used (see Section~\ref{subsec:tester-map-model}), then the neighbor of $(u,i)$ in the final graph is fixed, and we do not know how to efficiently emulate the corresponding (lazy) random walks by only using degree and neighbor queries to $G$.

Now we briefly introduce a process for emulating random walks on $G'$. The argument is very similar to the analogous case given in Section 4.2 in \cite{KKR04:bipartite}. We give a brief description here. To emulate random walks on $G'$, if we are currently at a vertex $(u,i)$, then with probability $1/2$, we stay at $(u,i)$; with probability $1/4$, we jump to a randomly chosen neighbor $(u,j)$ in $H_u$, which can be done in constant time since $H_u$ is explicitly constructible; with the remaining probability $1/4$, we need to jump to the outside of $H_u$. Now if we have already specified its neighbor outside of $H_u$, say $(v,j)$, then we directly jump to $(v,j)$. Otherwise, we have to specify the outside neighbor of $(u,i)$ first. The specification can be done by recording a set $A(u)$ of neighbors that has already been specified to some vertex in $H_u$ and then either sampling new neighbors or attaching unspecified vertices arbitrarily according to $A(u)$. The amortized number of required degree and neighbor queries to $G$ is $O(\log^2 n)$. We refer to~\cite{KKR04:bipartite} for more details.

There is one more issue that we should take care of: how to sample vertices (almost) uniformly at random from $G'$. This issue is almost equivalent to sampling edges almost uniformly from $G$, and has also been analyzed in~\cite{KKR04:bipartite}. In particular, Kaufman et al. have proved the following lemma.
\begin{lemma}[\cite{KKR04:bipartite}]~\label{lem:sample}
Let $\mu>0$. There exists a procedure \texttt{Sample\--Edges\--Almost\--Uniformly\--in\--$G$} that performs $O(\sqrt{n/\mu}\log m)$ degree and neighbor queries and for all but $(\mu/4)m$ of edges $e$ in $G$, the probability that the procedure outputs $e$ is at least $1/(64m)$. In particular, the output edge $e$ is in the form of $(v,i)$ for $1\leq i\leq \deg(v)$.
\end{lemma}

By setting $\mu=\varepsilon/c'$ in the above lemma, for a sufficiently large constant $c'$, we will directly invoke \texttt{Sample-Edges-Almost-Uniformly-in-$G$} to sample a vertex $(v,i)$ in $G'$.

Finally, to specify the number of random walks $r$, to be $O(\sqrt{m})$, we should have an estimate of $m$ or the average degree $d_{avg}$ of $G$. This can be achieved by Feige's algorithm~\cite{Fei06:sum,GR08:average}, which gives a constant factor estimate of $d_{avg}$ by performing $O(\sqrt{n})$ queries to $G$.

Now we give a description of our two-sided error tester.
\begin{center}
\begin{tabular}{|p{0.8\textwidth}|}
\hline
\texttt{SSETester2-List}$(G,s,r,\ell,\sigma)$
\begin{enumerate}
\item Repeat $s$ times:
\begin{enumerate}
\item Sample an edge $(v,i)$ by calling the procedure \texttt{Sample\--Edges\--Almost\--Uniformly\--in\--$G$} with $\mu=\varepsilon/c'$, where $c'$ is a sufficiently large constant.
\item Perform $r$ independent lazy random walks in $G_{\geq 8}\R\HH$ of length $\ell$ starting from $v$ by the above emulation process.
\item Let $Z_v$ be the number of pairwise collisions among the endpoints of these $r$ random walks.
\item If $Z_v>\sigma$ then abort and output \textbf{reject}.
\end{enumerate}

\item Output \textbf{accept}.
\end{enumerate}
\\
\hline
\end{tabular}
\end{center}

Now we are ready to prove Theorem~\ref{thm:twoside-list} by using similar analysis to the proof of Theorem~\ref{thm:bounded}
\begin{proof}[Proof of Theorem~\ref{thm:twoside-list}]
We set $s=\Theta(1/\varepsilon)$, $r=\Theta(\sqrt{m}/\varepsilon)$, $\ell=\Theta(\frac{\min\{\log(4m/k),\log n\}\cdot(\ln k)}{\phi^4})$ and $\sigma=\binom{r}{2}\frac{6c}{k\varepsilon}$ in the algorithm \texttt{SSETester2-List}. Let $k^*=\Theta(k\varepsilon)$ and $\phi^*=\Theta(\frac{\phi^4}{\min\{\log(4m/k),\log n\}\cdot(\ln k)})$.

The proof follows by combining the arguments in the proof of Theorem~\ref{thm:bounded} and our description of the implementation of \texttt{SSETester2-List}. We sketch the main idea below.

First note that we are actually testing the expansion profile of $G'=G_{\geq 8}\R\HH$, and that the number of vertices and edges in $G'$ are both $\Theta(m)$.

In the completeness of the tester, since $\phi_G(k)\geq \phi$, then by the definition of $G'$ and Lemma~\ref{lem:preserve-expansion}, $\phi_{G'}(\Theta(k))\geq \Omega(\phi^2)$. Furthermore, by applying $H=G'$, $n=|V(G')|=\Theta(m)$ and $m_H=|E(G')|=\Theta(m)$ in the proof of Lemma~\ref{lem:completeness-bounded}, we know that for at least $1-\frac{\varepsilon}{c}$ fraction of nodes $(u,i)$ in $G'$, the collision probability of random walk distribution of length $t=\ell$ from $(u,i)$ is at most $\frac{4c}{k\varepsilon}$, this will ensure that with probability at least $5/6$, for all sampled  $s=\Theta(1/\varepsilon)$ vertices $(u,i)$, the collision probabilities of corresponding random walks are at most $\frac{4c}{k\varepsilon}$. Then the correctness of the tester can be proven by similar arguments as for the proof of Lemma~\ref{lem:completeness-bounded} by replacing $\phi$ by $\Omega(\phi^2)$ there.


In the soundness part, if $G$ is $\varepsilon$-far from any $(k^*,\phi^*)$-expander, then $G_{\geq 8}$ is $\varepsilon$-far from any $(8k^*,\phi^*)$-expander. By Lemma~\ref{lem:partition}, there exist disjoint sets $S_1,\cdots,S_q$ such that $\vol_{G_{\geq 8}}(S_1\cup\cdots\cup S_q)\geq \frac{\varepsilon m}{15}$, and for each $i\leq q$, $\vol_{G_{\geq 8}}(S_i)\leq 2k^*$, $\phi_{G_{\geq 8}}(S_i)<11c_2\phi^*$. By Lemma~\ref{lem:preserve-expansion}, this implies that there exists disjoint sets $S_i'\subseteq V(G')$ satisfying that $|S_i'|= \vol_{G_{\geq 8}}(S_i)$, and $\phi_{G'}(S_i')\leq 11c_2\phi^*/2$. Now we apply Lemma~\ref{lem:local} with $G=G', S=S_i'$ and $t=\ell, \psi=11c_2\phi^*/2$ to find  $\widehat{S_i'}\subseteq S_i'$ such that $|\widehat{S_i'}|\geq \frac{1}{2}|S_i'|$. Note that since $\vol(S_1\cup\cdots\cup S_q)\geq \frac{\varepsilon m}{15}$, then $|\widehat{S_1'}\cup\cdots\cup \widehat{S_q'}|\geq \frac{1}{2}|S_1'\cup\cdots\cup S_q'|\geq \frac{\varepsilon m}{30}$. Now note that by Lemma~\ref{lem:sample}, the sampling procedure \texttt{Sample\--Edges\--Almost\--Uniformly\--in\--$G$} with $\mu=\varepsilon/c'$ for a sufficiently large constant $c'$, will output a vertex $(u,i)$ with probability at least $1/(64m)$, for all but $\frac{\varepsilon\cdot m}{4c'}$ vertices in $G'$. This further gives that by invoking \texttt{Sample\--Edges\--Almost\--Uniformly\--in\--$G$} for $s=\Theta(1/\varepsilon)$ times, we can guarantee that at least one vertex from $\widehat{S_1'}\cup\cdots\cup \widehat{S_q'}$ will be sampled out. This allows us to use the analogous arguments for the soundness of the tester for bounded-degree graphs to finish the proof.


Finally, note that the running time of the algorithm consists of the time to estimate the number of edges $m$, the time of calling  \texttt{Sample\--Edges\--Almost\--Uniformly\--in\--$G$}, the time of performing random walks and also estimating the collision probabilities. It is straightforward to see that the running time (and also query compleixty) is dominated by $O(r\ell s)=O(\frac{\sqrt{m}\min\{\log(4m/k),\log n\}\cdot(\ln k)}{\varepsilon^2\phi^4})$.
\end{proof}

\subsubsection{A one-sided error tester}~\label{sec:oneside}
Now we present our property testing algorithm \texttt{SSETester1-List} with one-sided error for small set expansion. This tester invokes a local algorithm  \texttt{LocalSS} introduced in Section~\ref{subsec:localalgorithm} and applies to the adjacency list model.
\begin{center}
\begin{tabular}{|p{0.8\textwidth}|}
\hline
\texttt{SSETester1-List}$(G,s,T,\delta)$
\begin{enumerate}
\item Repeat $s$ times:
\begin{enumerate}
\item Sample an edge $(v,i)$ by calling the procedure \texttt{Sample\--Edges\--Almost\--Uniformly\--in\--$G$} with $\mu=\varepsilon/c'$, where $c'$ is a sufficiently large constant.
\item If \texttt{LocalSS}$(G,v,T,\delta)$ finds a set $X$ with volume at most $k$ and conductance at most $\phi$, then abort and output \textbf{reject}.
\end{enumerate}

\item Output \textbf{accept}.
\end{enumerate}\\
\hline
\end{tabular}
\end{center}

Now we use the above algorithm to prove Theorem~\ref{thm:oneside}.
\begin{proof}[Proof of Theorem~\ref{thm:oneside}]
Let $\xi$ be $0<\xi<1$. In the algorithm \texttt{SSETester1-List}, we set $s=\Theta(1/\varepsilon)$, $T=O(\frac{\log k}{\phi^2})$, and $\delta=O(\frac{k^{-1+\xi/2}}{T})$. It is obvious that for any input graph $G$ that is a $(k,\phi)$-expander, \texttt{SSETester1-List} cannot output \textbf{reject}. Thus, we only need to consider the case that $G$ is $\varepsilon$-far from any $(k^*,\phi^*)$-expander, where $k^*=O(k^{1-\xi})$ and $\phi^*=O(\xi\phi^2)$ for any $0<\xi<1/2$. In this case, there exists disjoint subsets $S_1,\cdots,S_q\subseteq V$ with properties in Lemma~\ref{lem:partition}. Now for each $i\leq q$, by applying Lemma~\ref{lem:local} with $S=S_i,\psi=11c_2\phi^*,k=2k^*$ and $\zeta=\xi/2$, we know that there exists $\widehat{S_i}\subseteq S_i$ such that $\vol(\widehat{S_i})\geq\frac{1}{2}\vol(S_i)$, and that for each $v\in \widehat{S_i}$, the algorithm \texttt{LocalSS} with parameters $G, v, T=O(\frac{\xi\log k^*}{\xi\phi^2})=O(\frac{\log k}{\phi^2}), \delta=O(\frac{(k^*)^{-1-\xi/2}}{T})=O(\frac{k^{-1+\xi/2}}{T})$, will find a set $X$ such that $\vol(X)\leq O((k^*)^{1+\xi/2})=O(k^{1-\xi/2})<k$ and $\phi(X)\leq O(\sqrt{\psi/\xi})<\phi$ by our choice of $k^*$ and $\phi^*$. Finally, noting that $\vol(\widehat{S_1}\cup\cdots\cup \widehat{S_q})\geq\frac{1}{2}\vol(S_1\cup\cdots\cup S_q)\geq\frac{\varepsilon m}{30}$, and our sample size is $s=\Theta(1/\varepsilon)$, by the property of \texttt{Sample\--Edges\--Almost\--Uniformly\--in\--$G$} guaranteed in Lemma~\ref{lem:sample}, we can guarantee that with probability at least $2/3$, the algorithm will sample a vertex $v\in \widehat{S_1}\cup\cdots\cup \widehat{S_q}$ (and thus find a small non-expanding set) and then reject the graph $G$. Finally, note that the running time in each iteration is determined by the running time of subroutines \texttt{Sample\--Edges\--Almost\--Uniformly\--in\--$G$} and \texttt{LocalSS}. Then it is straightforward to see that the total running time of the algorithm is dominated by $\widetilde{O}(\sqrt{\frac{n}{\varepsilon^3}}+\frac{k}{\varepsilon \phi^{4}})$.
\end{proof}

\subsection{A tester in the rotation map model for general graphs}~\label{subsec:tester-map-model}
In this section, we give a tester in the rotation map model, in which we assume that the rotation map of $G$ is explicitly given, that is, when specified a vertex $v$ and an index $i$, the oracle returns a pair $(u,j)$ such that $u$ is the $i$th neighbor of $v$ and $j$ is the index of $u$ as a neighbor of $v$. We use the non-uniform replacement product with rotation map to transform $G$ into a $16$-regular graph $G'$. To perform this transformation, we also need first to turn $G$ into a graph $G_{\geq 8}$ with minimum degree $8$, and specify $\HH$ to be a proper $8$-regular Margulis expanders, and then let $G'=G_{\geq 8}^{(r)}\R\HH$. Now the tester first samples a number of vertices almost uniformly in $G'$ and then performs independent random walks on $G'$ to decide whether to accept $G$ or not, as we did before.

Our tester in rotation map model is almost the same as the two-sided tester in adjacency model in Section~\ref{subsec:two-sided-list}, and with information of the rotation map of $G$, we are actually able to give a better tester by using the spectral property of $G'$ given in Lemma~\ref{lem:separation} (see Theorem~\ref{thm:twoside-map}). However, as we mentioned before, since now we cannot specify the neighbor of a vertex $(u,i)$ in an arbitrary manner, we do not know how to emulate random walks efficiently by only performing degree and neighbor queries to $G$. That is why we introduced (neighbor,index) query and the rotation map model.

Here we emulate random walks on $G'$ by performing degree and (neighbor, index) queries to $G$: if we are currently at a vertex $(u,i)$, then with probability $1/2$, we stay at $(u,i)$; with probability $1/4$, we jump to a randomly chosen neighbor $(u,j)$ in $H_u$; with the remaining probability $1/4$, we jump to vertex $(v,j)$ such that $v$ is the $i$th neighbor of $u$ and $u$ is the $j$th neighbor of $v$ in $G$. Note that only in the last case, we need to perform (neighbor, index) queries to the oracle of $G$.
\begin{center}
\begin{tabular}{|p{0.8\textwidth}|}
\hline
\texttt{SSETester2-Map}$(G,s,r,\ell,\sigma)$
\begin{enumerate}
\item Repeat $s$ times:
\begin{enumerate}
\item Sample an edge $(v,i)$ by calling the procedure \texttt{Sample\--Edges\--Almost\--Uniformly\--in\--$G$} with $\mu = \varepsilon/c'$, where $c'$ is a sufficiently large constant.
\item Perform $r$ independent lazy random walks in $G_{\geq 8}^{(r)}\R H$ of length $\ell$ starting from $v$ by using rotation map of $G$.
\item Let $Z_v$ be the number of pairwise collisions among the endpoints of these $r$ random walks.
\item If $Z_v > \sigma$ then abort and output \textbf{reject}.
\end{enumerate}
\item Output \textbf{accept}.
\end{enumerate}
\\
\hline
\end{tabular}
\end{center}

Now we can use the above algorithm to prove Theorem~\ref{thm:twoside-map}.
\begin{proof}[Proof of Theorem~\ref{thm:twoside-map}]
We set $s=\Theta(1/\varepsilon)$, $r=\Theta(s\sqrt{m})$, $\ell=\Theta(\frac{\min\{\log n,\log(4m/k)\}\cdot(\ln k)}{\phi^2})$ and $\sigma=\binom{r}{2}\frac{6c}{k\varepsilon}$ in the algorithm \texttt{SSETester2-Map}. Let $k^*=\Theta(k\varepsilon)$ and $\phi^*=\Theta(\frac{\phi^2}{\min\{\log n,\log(4m/k)\}\cdot(\ln k)})$.

The proof is straightforward given the proof of Theorem~\ref{thm:bounded}, Theorem~\ref{thm:twoside-list} and our description of the implementation of \texttt{SSETester2-Map}.
A key difference is that now we directly use the spectral property of $G'$ that $\eta_{4m_{G'}/k}(W_{G'})\leq 1-\Omega(\frac{\phi^2}{\min\{\log(4m/k),\log n\}})$ to derive an upper bound for the collision probability, instead of using the combinatorial property that $\phi_{G'}(k)=\Omega(\phi^2)$ (which in turn gives that $\eta_{4m_{G'}/k}(W_{G'})\leq 1-\Omega(\frac{\phi^4}{\min\{\log(4m/k),\log n\}})$) as we did in the proof of Theorem~\ref{thm:twoside-list}. Here, the spectral property of $G'$ follows by its definition and Lemma~\ref{lem:separation}. Therefore, there is no quadratic loss of $\phi$ as we had in Theorem~\ref{thm:twoside-list}. The rest of the proof follows by analogous arguments in the proof of Theorem~\ref{thm:twoside-list}.
\end{proof}

\textbf{Remark.} We can use the ``Furthermore'' part of Lemma~\ref{lem:separation} to give stronger upper bound on the collision probability (of random walk distributions from each vertex) for $\phi$-expanders, which combined with the lower bound of collision probability given in~\cite{KS11:tester} can also be used to design conductance tester in rotation map model with the same running time and approximation guarantee as in~\cite{LPP11:conductance}. We omit the details here.


\section{Conclusions}\label{sec:conclusion}
We give property testers for small set expansion in general graphs, including a two-sided error tester and a one-sided error tester in adjacency list model, and a two-sided error tester in rotation map model in which the algorithm can perform (neighbor, index) queries as well as degree queries. Our analysis for two-sided error testers uses a non-uniform replacement product to transform an arbitrary graph into a bounded degree graph that well preserves expansion profile.

It is unclear if the rotation map model is strictly stronger than the adjacency list model. In particular, we do not know if the newly introduced (neighbor, index) query is necessary for us to obtain a tester with at most quadratic loss in the conductance parameter. It will be interesting to give a two-sided error tester in the adjacency list model that distinguishes $(k,\phi)$-expanders from graphs that are $\varepsilon$-far from any $(\Theta(k\varepsilon),\widetilde{\Theta}(\phi^2))$-expander, as we obtained in the rotation map model. It is also left open if the query complexity and/or running time of the two-sided testers could be improved to $\widetilde{O}(\sqrt{n}(\phi^{-1}\varepsilon^{-1})^{O(1)})$, without dependency on the number of edges $m$.

\subparagraph*{Acknowledgements}
We would like to thank anonymous referees of RANDOM 2014 for their very detailed and helpful comments to an earlier version of this paper.

\appendix






\bibliographystyle{alphabetic}
\bibliography{testssexpander}

\end{document}